\documentclass[11pt]{myclass}
\usepackage{euscript}
\newcommand{\eps}{\varepsilon}

\renewcommand{\c}[1]{\ensuremath{\EuScript{#1}}}
\renewcommand{\b}[1]{\ensuremath{\mathbb{#1}}}
\newcommand{\lD}{D}
\newcommand{\cD}{\text{disc}}

\newcommand{\plog}{\text{ polylog}}
\newcommand{\plle}{\ensuremath{\text{ polylog}(\log \frac{1}{\eps})}}

\newcommand{\nor}{\c{N}}
\newcommand{\erf}{\textsf{erf}}
\newcommand{\fra}[1]{\ensuremath{\downharpoonleft\hspace{-1.5mm} #1 \hspace{-1.5mm}\downharpoonright}}

\newcommand{\mypara}[1]{\paragraph{\textbf{\emph{#1}}}}

\title{Algorithms for $\eps$-approximations of Terrains\footnotemark[1]}
\author{Jeff M. Phillips\footnotemark[2]}

\begin{document}
\begin{titlepage}
\footnotetext[1]{$^*$Work on this paper is supported by a James B. Duke Fellowship,
 by NSF under a Graduate Research Fellowship and grants CNS-05-40347, CFF-06-35000, and DEB-04-25465, by ARO grants W911NF-04-1-0278 and W911NF-07-1-0376, by an NIH grant 1P50-GM-08183-01, by a DOE grant OEGP200A070505, and by a grant from the U.S. Israel Binational Science Foundation.}
\footnotetext[2]{$^\dagger$Department of Computer Science, Duke University, Durham, NC 27708: \texttt{jeffp@cs.duke.edu}}
\maketitle

\begin{abstract} 
Consider a point set $\c{D}$ with a measure function $\mu : \c{D} \to \b{R}$.
Let $\c{A}$ be the set of subsets of $\c{D}$ induced by containment in a shape from some geometric family (e.g. axis-aligned rectangles, half planes, balls, $k$-oriented polygons).  
We say a range space $(\c{D}, \c{A})$ has an $\eps$-approximation $P$ if
$$\max_{R \in \c{A}} \left| \frac{\mu(R \cap P)}{ \mu(P)} - \frac{\mu(R \cap \c{D})}{ \mu(\c{D})} \right| \leq \eps.$$

We describe algorithms for deterministically constructing discrete $\eps$-app- roximations for continuous point sets such as distributions or terrains. Furthermore, for certain families of subsets $\c{A}$, such as those described by axis-aligned rectangles, we reduce the size of the $\eps$-approximations by almost a square root from $O(\frac{1}{\eps^2} \log \frac{1}{\eps})$ to $O(\frac{1}{\eps} \plog \frac{1}{\eps})$.  
This is often the first step in transforming a continuous problem into a discrete one for which combinatorial techniques can be applied. We describe applications of this result in geo-spatial analysis, biosurveillance, and sensor networks.
\end{abstract}
\end{titlepage}

\section{Introduction}
Representing complex objects by point sets may require less storage and may make computation on them faster and easier.  When properties of the point set approximate those of the original object, then problems over continuous or piecewise-linear domains are now simple combinatorial problems over point sets.  
For instance, when studying terrains, representing the volume by the cardinality of a discrete point set transforms calculating the difference between two terrains in a region to just counting the number of points in that region.  Alternatively, if the data is already a discrete point set, approximating it with a much smaller point set has applications in selecting sentinel nodes in sensor networks.  
This paper studies algorithms for creating small samples with guarantees in the form of discrepancy and $\eps$-approximations, in particular we construct $\eps$-approximations of size $O(\frac{1}{\eps} \plog \frac{1}{\eps})$.  

\mypara{$\eps$-approximations.}
In this paper we study point sets, which we call domains and we label as $\c{D}$, which are either finite sets or are Lebesgue-measureable sets.  For a given domain $\c{D}$ let $\c{A}$ be a set of subsets of $\c{D}$ induced by containment in some geometric shape (such as balls or axis-aligned rectangles).  The pair $(\c{D}, \c{A})$ is called a \emph{range space}.  
We say that $P$ is an \emph{$\eps$-approximation} of $(\c{D}, \c{A})$ if
$$\max_{R \in \c{A}} \left| \frac{|R \cap P|}{|P|} - \frac{|R \cap \c{D}|}{|\c{D}|} \right| \leq \eps,$$
where $| \cdot |$ represents the cardinality of a discrete set or the Lebesgue measure for a Lebesgue-measurable set.  
$\c{A}$ is said to \emph{shatter} a discrete set $X \subseteq \c{D}$ if each subset of $X$ is equal to $R \cap X$ for some $R \in \c{A}$.  The cardinality of the largest discrete set $X$ that $\c{A}$ can shatter is known as the \emph{VC-dimension}.  
A classic result of Vapnik and Chervonenkis~\cite{VC71} states that for any range space $(\c{D},\c{A})$ with constant VC-dimension $v$ there exists a subset $P \subset \c{D}$ consisting of $O(\frac{v}{\eps^2} \log \frac{v}{\eps})$ points that is an $\eps$-approximation for $(\c{D},\c{A})$.  Furthermore, if each element of $P$ is drawn uniformly at random from $\c{D}$ such that $|P| = O(\frac{v}{\eps^2} \log \frac{v}{\eps \delta})$, then $P$ is an $\eps$-approximation with probability at least $1-\delta$.  
Thus, for a large class of range spaces random sampling produces an $\eps$-approximation of size $O(\frac{1}{\eps^2} \log \frac{1}{\eps})$.

\mypara{Deterministic construction of $\eps$-approximations.}
There exist deterministic constructions for $\eps$-approximations.  When $\c{D}$ is the unit cube $[0,1]^d$ there are constructions which can be interpreted as $\eps$-approximations of size $O(\frac{1}{\eps^{2d/(d+1)}})$ for half spaces~\cite{Mat95} and $O(\frac{1}{\eps^{2d/(d+1)}} \log^{d/(d+1)} \frac{1}{\eps}\plle)$ for balls in $d$-dimensions~\cite{Bec87}.  Both have lower bounds of $\Omega(\frac{1}{\eps^{2d/(d+1)}})$~\cite{Ale91}. See Matou\v{s}ek~\cite{Mat99} for more similar results or Chazelle's book \cite{Cha00} for applications.  
For a domain $\c{D}$, 
let $\c{R}_d$ describe the subsets induced by axis-parallel rectangles in $d$ dimensions, and let $\c{Q}_k$ describe the subsets induced by $k$-oriented polygons (or more generally polytopes) with faces described by $k$ predefined normal directions.  
More precisely, for $\beta = \{\beta_1, \ldots, \beta_k\} \subset \b{S}^{d-1}$, let $\c{Q}_\beta$ describe the set of convex polytopes such that each face has an outward normal $\pm \beta_i$ for $\beta_i \in \beta$.  
If $\beta$ is fixed, we will use $\c{Q}_k$ to denote $\c{Q}_\beta$ since it is the size $k$ and not the actual set $\beta$ that is important.
When $\c{D} = [0,1]^d$, then the range space $(\c{D}, \c{R}_d)$ has an $\eps$-approximation of size $O(\frac{1}{\eps} \log^{d-1} \frac{1}{\eps}\plle)$~\cite{Hal60}.  
Also, for all homothets (translations and uniform scalings) of any particular $Q \in \c{Q}_k$, Skriganov constructs an $\eps$-approximation of size $O(\frac{1}{\eps} \log^{d-1} \frac{1}{\eps}\plle)$.  
When $\c{D}$ is a discrete point set of size $n$, $\eps$-approximations of size $O((\frac{1}{\eps} \log \frac{1}{\eps})^{2-\frac{2}{v+1}})$ exist for bounded VC-dimension $v$~\cite{MWW93}, and can be constructed in time $O(n \cdot \frac{1}{\eps^{2v}} \log^v \frac{1}{\eps})$.
In this spirit, for $\c{R}_2$ and a discrete point set of size $n$, Suri, Toth, and Zhou \cite{STZ04} construct an $\eps$-approximation of size $O(\frac{1}{\eps} \log(\eps n) \log^4 (\frac{1}{\eps} \log (\eps n)))$ in the context of a streaming algorithm which can be analyzed to run in time $O(n (\frac{1}{\eps} \log^4 \frac{1}{\eps})^3)$.

\mypara{Our results.}
We answer the question, ``\emph{for which ranges spaces can we construct $\eps$-approximations of size} $O(\frac{1}{\eps} \plog \frac{1}{\eps})$?''  
 by describing how to deterministically construct an $\eps$-approximation of size $O(\frac{1}{\eps} \plog \frac{1}{\eps})$ for any domain which can be decomposed into or approximated by a finite set of constant-size polytopes for families $\c{R}_d$ and $\c{Q}_k$.  
In particular:
\begin{itemize}
\item 
For a discrete point set $\c{D}$ of cardinality $n$, we give an algorithm for generating an $\eps$-approximation for $(\c{D}, \c{Q}_k)$ of size $O(\frac{1}{\eps} \log^{2k} \frac{1}{\eps} \plle)$ in $O(n\frac{1}{\eps^3} \plog \frac{1}{\eps})$ time.  
This requires a generalization of the iterative point set thinning algorithm by Chazelle and Matou\v{s}ek~\cite{CM96} that does not rely on VC-dimension.  
This implies similar results for $\c{R}_d$ as well.
\item 
For any $d$-dimensional domain $\c{D}$ that can be decomposed into $n$ $k^\prime$-oriented polytopes, we give an algorithm for generating an $\eps$-approximation of size $O((k+k^\prime) \frac{1}{\eps} \log^{2k}\frac{1}{\eps}\plle)$ for $(\c{D}, \c{Q}_k)$ in time $O((k+k^\prime) n \frac{1}{\eps^4} \plog \frac{1}{\eps})$.  
\end{itemize}
We are interested in terrain domains $\c{D}$ defined to have a base $B$ (which may, for instance, be a subset of $\b{R}^2$) and a height function $h: B \to \b{R}$.  Any point $(p, z)$ such that $p \in B$ and $0 \leq z \leq h(p)$ (or $0 \geq z \geq h(p)$ when $h(p) < 0$) is in the domain $\c{D}$ of the terrain.  
\begin{itemize}
\item 
For a terrain domain $\c{D}$ where $B$ and $h$ are piecewise-linear with $n$ linear pieces, our result implies that there exists an $\eps$-approximation of size $O(k \frac{1}{\eps} \log^4 \frac{1}{\eps}\plle)$ for $(\c{D}, \c{Q}_k)$, and it can be constructed in $O(n \cdot \frac{1}{\eps^4} \plog \frac{1}{\eps})$ time.  
\item 
For a terrain domain $\c{D}$ where $B \subset \b{R}^2$ is a rectangle with diameter $d$ and $h$ is smooth ($C^2$-continuous) with minimum height $z^-$ and largest eigenvalue of its Hessian $\lambda$, we give an algorithm for creating an $\eps$-approximation for $(\c{D}, \c{R}_2 \times \b{R})$ of size $O(\frac{1}{\eps} \log^{4} \frac{1}{\eps}\plle)$ in time $O(\frac{\lambda d^2}{z^-} \frac{1}{\eps^5} \plog \frac{1}{\eps})$.  
\end{itemize}

These results improve the running time for a spatial anomaly detection problem in biosurveillance~\cite{AMPVZ06},   
and can more efficiently place or choose sentinel nodes in a sensor network, addressing an open problem~\cite{SST05}.

\mypara{Roadmap.}
We introduce a variety of new techniques, rooted in discrepancy theory, to create $\eps$-approximations of size $O(\frac{1}{\eps} \plog \frac{1}{\eps})$ for increasingly difficult domains.  
First, Section \ref{sec:l+c-disc} discusses Lebesgue and combinatorial discrepancy.
Section \ref{sec:compressing} generalizes and improves a classic technique to create an $\eps$-approximation for a discrete point set.
Section \ref{sec:sample-poly} describes how to generate an $\eps$-approximation for a polygonal domain.  
When a domain can be decomposed into a finite, disjoint set of polygons, then each can be given an $\eps$-approximation and the union of all these point sets can be given a smaller $\eps$-approximation using the techniques in Section \ref{sec:compressing}.  
Section \ref{sec:sample-smooth} then handles domains of continuous, non-polygonal point sets by first approximating them by a disjoint set of polygons and then using the earlier described techniques.  
Section \ref{sec:app} shows some applications of these results.

\section{Lebesgue and Combinatorial Discrepancy}
\label{sec:l+c-disc}

\mypara{Lebesgue discrepancy.}
The Lebesgue discrepancy is defined for an $n$-point set $P \subset [0,1]^d$ relative to the volume of a unit cube $[0,1]^d$.  \footnote{Although not common in the literature, this definition can replace $[0,1]^d$ with an hyper-rectangle $[0,w_1] \times [0,w_2] \times \ldots \times [0,w_d]$.}  Given a range space $([0,1]^d,\c{A})$ and a point set $P$, the \emph{Lebesgue} discrepancy is defined
$$\lD(P,\c{A}) = \sup_{R \in \c{A}} |\lD(P,R)|, \; \; \textrm{ where } \; \; \lD(P,R) = n \cdot |R \cap [0,1]^d| - |R \cap P|.$$
Optimized over all $n$-point sets, define the \emph{Lebesgue discrepancy of $([0,1]^d, \c{A})$} as
$$\lD(n,\c{A}) = \inf_{P \subset[0,1]^d, |P|=n} \lD(P,\c{A}).$$

The study of Lebesgue discrepancy arguably began with the Van der Corput set $C_n$~\cite{vdC35}, which satisfies $\lD(C_n, \c{R}_2) = O(\log n)$.  This was generalized to higher dimensions by Hammersley~\cite{Ham60} and Halton~\cite{Hal60} so that $\lD(C_n, \c{R}_d) = O(\log^{d-1} n)$.    However, it was shown that many lattices also provide $O(\log n)$ discrepancy in the plane~\cite{Mat99}.  This is generalized to $O(\log^{d-1} n \log^{1+\tau} \log n)$ for $\tau >0$ over $\c{R}^d$~\cite{Skr90,Skr95,Bec94}.  For a more in-depth history of the progression of these results we refer to the notes in Matou\v{s}ek's book~\cite{Mat99}.  For application of these results in numerical integration see Niederreiter's book~\cite{Nie92}.  
The results on lattices extend to homothets of any $Q_k \in \c{Q}_k$ for $O(\log n)$ discrepancy in the plane~\cite{Skr90} and $O(\log^{d-1} n \log^{1+\tau} \log n)$ discrepancy, for $\tau>0$, in $\b{R}^d$~\cite{Skr98}, for some constant $k$.  
A wider set of geometric families which include half planes, right triangles, rectangles under all rotations, circles, and predefined convex shapes produce $\Omega(n^{1/4})$ discrepancy and are not as interesting from our perspective.

Lebesgue discrepancy describes an $\eps$-approximation of $([0,1]^d, \c{A})$, where $\eps  = f(n) = D(n, \c{A})/n$.  Thus we can construct an $\eps$-approximation for $([0,1]^d, \c{A})$ of size $g_\c{\lD}(\eps, \c{A})$ as defined below. (Solve for $n$ in $\eps = D(n,\c{A})/n)$.)
\begin{equation}
\label{eq:g-LD-eps}
g_{\lD}(\eps, \c{A}) = 
\begin{cases} 
  O(\frac{1}{\eps} \log^\tau \frac{1}{\eps} \plog (\log \frac{1}{\eps})) & \textrm{ for } \lD(n, \c{A}) = O(\log^\tau n)\\
  O((1 / \eps)^{1/(1-\tau)}) & \textrm{ for }  \lD(n, \c{A}) = O(n^\tau)
\end{cases}
\end{equation}

\mypara{Combinatorial discrepancy.}
Given a range space $(X, \c{A})$ where $X$ is a finite point set and a coloring function $\chi : X \to \{-1, +1\}$ we say the \emph{combinatorial discrepancy} of $(X,\c{A})$ colored by $\chi$ is 
$$
\cD_{\chi}(X,\c{A}) = \max_{R \in \c{A}} \cD_{\chi}(X \cap R) \;\; \textrm{ where }
$$
$$
\cD_{\chi}(X) = \sum_{x \in X} \chi(x) = \left| \{ x \in X \;:\; \chi(x) = +1 \} \right| - \left| \{ x \in X \;:\; \chi(x) = -1 \} \right|.
$$
Taking this over all colorings and all point sets of size $n$ we say
$$\cD(n, \c{A}) = \max_{\{X : |X| = n\}} \min_{\chi : X \to \{-1,+1\}} \cD_{\chi}(X,\c{A}).$$

Results about combinatorial discrepancy are usually proved using the partial coloring method~\cite{Bec81-roth} or the Beck-Fiala theorem~\cite{BF81}.  The partial coloring method usually yields lower discrepancy by some logarithmic factors, but is nonconstructive.  Alternatively, the Beck-Fiala theorem actually constructs a low discrepancy coloring, but with a slightly weaker bound.  The Beck-Fiala theorem states that 
%
for a family of ranges $\c{A}$ and a point set $X$ such that $\max_{x \in X} |\{A \in \c{A} \;:\; x \in A\}| \leq t$,  $ \cD(X, \c{A}) \leq 2t-1$.
So the discrepancy is only a constant factor larger than the largest number of sets any point is in.  

Srinivasan \cite{Sri97} shows that $\cD(n, \c{R}_2) = O(\log^{2.5}n)$, using the partial coloring method.  An earlier result of Beck \cite{Bec81} showed $\cD(n, \c{R}_2) = O(\log^4 n)$ using the Beck-Fiala theorem \cite{BF81}.  
The construction in this approach reduces to $O(n)$ Gaussian eliminations on a matrix of constraints that is $O(n) \times O(n)$.  Each Gaussian elimination step requires $O(n^3)$ time. 
Thus the coloring $\chi$ in the construction for $\cD(n, \c{R}_2) = O(\log^4 n)$ can be found in $O(n^4)$ time.
We now generalize this result.



\begin{lemma}
\label{thm:cd-pol}
$\emph{\cD}(n,\c{Q}_k) = O(\log^{2k} n)$ for points in $\b{R}^d$ and the coloring that generates this discrepancy can be constructed in $O(n^4)$ time, for $k$ constant.  
\end{lemma}
The proof combines techniques from Beck \cite{Bec81} and Matou\v{s}ek \cite{Mat00}.  
\begin{proof}
Given a class $\c{Q}_k$, each potential face is defined by a normal vector from $\{\beta_1, \ldots, \beta_k\}$.  For $j \in [1,k]$ project all points along $\beta_j$.  
Let a \emph{canonical interval} be of the form $\left[ \frac{t}{2^q}, \frac{t+1}{2^q} \right)$ for integers $q \in [1,\log n]$ and $t \in [0, 2^q)$.  
For each direction $\beta_j$ choose a value $q \in [1,\log n]$ creating $2^q$ canonical intervals induced by the ordering along $\beta_j$.  Let the intersection of any $k$ of these canonical intervals along a fixed $\beta_j$ be a \emph{canonical subset}.  Since there are $\log n$ choices for the values of $q$ for each of the $k$ directions, it follows that each point is in at most $(\log n)^k$ canonical subsets.  Using the Beck-Fiala theorem, we can create a coloring for $X$ so that no canonical subset has discrepancy more than $O(\log^k n)$.  

Each range $R \in \c{Q}_k$ is formed by at most $O(\log^k n)$ canonical subsets.  For each ordering by $\beta_i$, the interval in this ordering induced by $R$ can be described by $O(\log n)$ canonical intervals.  Thus the entire range $R$ can be decomposed into $O(\log^k n)$ canonical subsets, each with at most $O(\log^k n)$ discrepancy.  

Applying the Beck-Fiala construction of size $n$, this coloring requires $O(n^4)$ time to construct.
\end{proof}

\begin{corollary}
\label{thm:cd-rectd}
$\cD(n, \c{R}_d) = O(\log^{2d} n)$ and the coloring that generates this discrepancy can be constructed in $O(n^4)$ time, for $d$ constant.  
\end{corollary}


A better nonconstructive bound exists due to Matou\v{s}ek~\cite{Mat00}, using the partial coloring method.  For polygons in $\b{R}^2$ $\cD(n, Q_k) = O(k \log^{2.5} n \sqrt{\log(k + \log n)})$, and for polytopes in $\b{R}^d$ $\cD(n,Q_k) = O(k^{1.5 \lfloor d/2 \rfloor} \log^{d + 1/2} n \sqrt{\log ( k + \log n)})$.  
For more results on discrepancy see Beck and Chen's book \cite{BC87}.

Similar to Lebesgue discrepancy, the set $P = \{p \in X \mid \chi(p) = +1\}$ generated from the coloring $\chi$ for combinatorial discrepancy $\cD(n, \c{A})$ describes an $\eps$-approximation of $(X,\c{A})$ where $\eps  = f(n) = \cD(n, \c{A})/n$.  Thus, given this value of $\eps$, we can say that $P$ is an $\eps$-approximation for $(X, \c{A})$ of size 
\begin{equation}  
\label{eq:g-CD-eps}
g(\eps, \c{A}) = 
\begin{cases} 
  O(\frac{1}{\eps} \log^\tau \frac{1}{\eps} \plog (\log \frac{1}{\eps})) & \textrm{ for } \cD(n, \c{A}) = O(\log^\tau n)\\
  O((1 / \eps)^{1/(1-\tau)}) & \textrm{ for }  \cD(n, \c{A}) = O(n^\tau).
\end{cases}
\end{equation}
The next section will describe how to iteratively apply this process efficiently to achieve these bounds for any value of $\eps$.

\section{Deterministic Construction of $\eps$-approximations for Discrete Point Sets}
\label{sec:compressing}




We generalize the framework of Chazelle and Matou\v{s}ek \cite{CM96} describing an algorithm for creating an $\eps$-approximation of a range space $(X, \c{A})$.  Consider any range space $(X,\c{A})$, with $|X|=n$, for which there is an algorithm to generate a coloring $\chi$ that yields the combinatorial discrepancy $\cD_{\chi}(X,\c{A})$ and can be constructed in time $O(n^w \cdot l(n))$ where $l(n) = o(n)$.  
For simplicity, we refer to the combinatorial discrepancy we can construct $\cD_{\chi}(X,\c{A})$ as $\cD(n,\c{A})$ to emphasize the size of the domain, 
and we use equation (\ref{eq:g-CD-eps}) to describe $g(\eps, \c{A})$, the size of the $\eps$-approximation it corresponds to.  
The values $\cD(n,\c{A})$, $w$, and $l(n)$ are dependent on the family $\c{A}$ (e.g. see Lemma \ref{thm:cd-pol}), but not necessarily its VC-dimension as in \cite{CM96}.  
As used above, let $f(n) = \cD(n,\c{A})/n$ be the value of $\eps$ in the $\eps$-approximation generated by a single coloring of a set of size $n$ --- the relative error.
We require that, $f(2n) \leq (1-\delta) f(n)$, for constant $0 < \delta \leq 1$; thus it is a geometrically decreasing function.  

The algorithm will compress a set $X$ of size $n$ to a set $P$ of size $O(g(\eps, \c{A}))$ such that $P$ is an $\eps$-approximation of $(X,\c{A})$ by recursively creating a low discrepancy coloring.  
We note that an $\eps$-approximation of an $\eps^\prime$-approximation is an $(\eps + \eps^\prime)$-approximation of the original set. 

We start by dividing $X$ into sets of size $O(g(\eps, \c{A}))$,\footnote{If the sets do not divide equally, artificially increase the size of the sets when necessary.  These points can be removed later.} here $\eps$ is a parameter.
The algorithm proceeds in two stages.  
The first stage alternates between merging pairs of sets and halving sets by discarding points colored $\chi(p) = -1$ by the combinatorial discrepancy method described above.  The exception is after every $w+2$ halving steps, we then skip one halving step. 
The second stage takes the one remaining set and repeatedly halves it until the error $f(|P|)$ incurred in the remaining set $P$ exceeds $\frac{\eps}{2 + 2\delta}$.  This results in a set of size $O(g(\eps, \c{A}))$.  

\begin{algorithm}[h!!t]
\caption{\label{alg:trim} Creates an $\eps$-approximation for $(X, \c{A})$ of size $O(g(\eps, \c{A}))$.}
\begin{algorithmic}[1]
\STATE Divide $X$ into sets $\{X_0, X_1, X_2, \ldots \}$ each of size $4(w+2) g(\eps, \c{A})$. \footnotemark[2]
\REPEAT[\emph{Stage 1}]
\FOR[or stop if only one set is left] {$w+2$ steps} 
\STATE \textsc{Merge: } Pair sets arbitrarily (i.e. $X_i$ and $X_j$) and merge them into a single set (i.e. $X_i := X_i \cup X_j$).
\STATE \textsc{Halve: } Halve each set $X_i$ using the coloring $\chi$ from $\cD(X_i, \c{A})$ (i.e. $X_i = \{x \in X_i \mid \chi(x) = +1\}$).
\ENDFOR
\STATE \textsc{Merge: } Pair sets arbitrarily and merge each pair into a single set.
\UNTIL {only one set, $P$, is left}
\REPEAT[\emph{Stage 2}]
\STATE \textsc{Halve: } Halve $P$ using the coloring $\chi$ from $\cD(P, \c{A})$.
\UNTIL {$f(|P|) \geq \eps / (2 + 2 \delta)$}
\end{algorithmic}
\end{algorithm}

\begin{theorem}
\label{thm:thin}
For a finite range space $(X, \c{A})$ with $|X|=n$ and an algorithm to construct a coloring $\chi : X \to \{-1, +1\}$ such that
\begin{itemize}
\item the set $\{x \in X \;:\; \chi(x) = +1\}$ is an $\alpha$-approximation of $(X, \c{A})$ of size $g(\alpha, \c{A})$ with $\alpha = \cD_{\chi}(X, \c{A})/ n$ (see equation (\ref{eq:g-CD-eps})).
\item $\chi$ can be constructed in $O(n^w \cdot l(n))$ time where $l(n) = o(n)$. 
\end{itemize}
then Algorithm \ref{alg:trim} constructs an $\eps$-approximation for $(X,\c{A})$ of size $O(g(\eps, \c{A}))$ in time $O(w^{w-1} n \cdot g(\eps, \c{A})^{w-1} \cdot l(g(\eps, \c{A})) + g(\eps, \c{A}))$.  
\end{theorem}

\begin{proof}
Let $2^j = 4(w+2) g(\eps, \c{A})$, for an integer $j$,  be the size of each set in the initial dividing stage (adjusting by a constant if $\delta \leq \frac{1}{4}$).
Each round of Stage 1 performs $w+3$ \textsc{Merge} steps and $w+2$ \textsc{Halve} steps on sets of the same size and each subsequent round deals with sets twice as large.  
The union of all the sets is an $\alpha$-approximation of $(X,\c{A})$ (to start $\alpha=0$) and $\alpha$ only increases in the \textsc{Halve} steps.  The $i$th round increases $\alpha$ by $f(2^{j-1+i})$ per \textsc{Halve} step.  Since $f(n)$ decrease geometrically as $n$ increases, the size of $\alpha$ at the end of the first stage is asymptotically bounded by the increase in the first round.  Hence, after Stage 1 $\alpha \leq 2(w+2)f(4(w+2)g(\eps, \c{A})) \leq \frac{\eps}{2}$.  
Stage 2 culminates the step before $f(|P|) \geq \frac{\eps}{2 + 2\delta}$.  
Thus the final \textsc{Halve} step creates an $\frac{\eps \delta}{2 + 2\delta}$-approximation and the entire second stage creates an $\frac{\eps}{2}$-approximation, hence overall Algorithm \ref{alg:trim} creates an $\eps$-approximation.  
The relative error caused by each \textsc{Halve} step in stage 2 is equivalent to a \textsc{Halve} step in a single round of stage 1.  

The running time is also dominated by Stage 1.  Each \textsc{Halve} step of a set of size $2^j$ takes $O((2^j)^w l(2^j))$ time and runs on $n / 2^j$ sets.  In between each \textsc{Halve} step within a round, the number of sets is divided by two, so the running time is asymptotically dominated by the first \textsc{Halve} step of each round.  The next round has sets of size $2^{j+1}$, but only $n/2^{j+w+2}$ of them, so the runtime is at most $\frac{1}{2}$ that of the first \textsc{Halve} step.  Thus the running time of a round is less than half of that of the previous one.  Since $2^j = O(w g(\eps, \c{A}))$ the running time of the \textsc{Halve} step, and hence the first stage is bounded by $O(n \cdot (w \cdot g(\eps,\c{A}))^{w-1} \cdot l(g(\eps, \c{A})) + g(\eps, \c{A}))$.  
Each \textsc{Halve} step in the second stage corresponds to a single \textsc{Halve} step per round in the first stage, and does not affect the asymptotics. 
\end{proof}

We can invoke Theorem \ref{thm:thin} along with Lemma \ref{thm:cd-pol} and Corollary \ref{thm:cd-rectd} to compute $\chi$ in $O(n^4)$ time (notice that $w=4$ and $l(\cdot)$ is constant), so $g(\eps, \c{Q}_k) = O(\frac{1}{\eps} \log^{2k} \frac{1}{\eps} \plle)$ and $g(\eps, \c{R}_d) = O(\frac{1}{\eps} \log^{2d} \frac{1}{\eps} \plle)$.  
We obtain the following important corollaries.

\begin{corollary}
\label{cor:compress-Qk}
For a set of size $n$ and over the ranges $\c{Q}_k$ an $\eps$-approximation of size $O(\frac{1}{\eps} \log^{2k} \frac{1}{\eps}\emph{\plle})$ can be constructed in time $O(n\frac{1}{\eps^3} \emph{\plog} \frac{1}{\eps})$.
\end{corollary}

\begin{corollary}
\label{cor:compress-Rd}
For a set of size $n$ and over the ranges $\c{R}_d$ an $\eps$-approximation of size $O(\frac{1}{\eps} \log^{2d} \frac{1}{\eps}\emph{\plle})$ can be constructed in time $O(n\frac{1}{\eps^3} \emph{\plog} \frac{1}{\eps})$.
\end{corollary}

\mypara{Weighted case.}
These results can be extended to the case where each point $x \in X$ is given a weight $\mu(x)$.  Now an $\eps$-approximation is a set $P \subset X$ and a weighting $\mu : X \to \b{R}$ such that
$$\max_{R \in \c{A}} \left| \frac{\mu(P \cap R)}{\mu(P)} - \frac{\mu(X \cap R)}{\mu(X)} \right|  \leq \eps,$$
where $\mu(P) = \sum_{p \in P} \mu(p)$.  
The weights on $P$ may differ from those on $X$.  
A result from Matou\v{s}ek \cite{Mat91}, invoking the unweighted algorithm several times at a geometrically decreasing cost, creates a weighted $\eps$-approximation of the same asymptotic size and with the same asymptotic runtime as for an unweighted algorithm.  
This extension is important when we combine $\eps$-approximations representing regions of different total measure.  For this case we weight each point relative to the measure it represents.  

\section{Sampling from Polygonal Domains}
\label{sec:sample-poly}
We will prove a general theorem for deterministically constructing small $\eps$-approximations for polygonal domains which will have direct consequences on polygonal terrains.  A key observation of Matou\v{s}ek \cite{Mat91} is that the union of $\eps$-approximations of disjoint domains forms an $\eps$-approximation of the union of the domains.  
Thus for any geometric domain $\c{D}$ we first divide it into pieces for which we can create $\eps$-approximations.  Then we merge all of these point sets into an $\eps$-approximation for the entire domain.  Finally, we use Theorem \ref{thm:thin} to reduce the sample size.  

Instead of restricting ourselves to domains which we can divide into cubes of the form $[0,1]^d$, thus allowing the use of Lebesgue discrepancy results, we first expand on a result about lattices and polygons.

\mypara{Lattices and polygons.}
For $x \in \b{R}$, let $\fra{x}$ represent the fractional part of $x$, and for $\alpha \in \b{R}^{d-1}$ let $\alpha  = ( \alpha_1, \ldots, \alpha_{d-1} )$.  
Now given $\alpha$ and $m$ let $P_{\alpha, m} = \{p_0, \ldots, p_{m-1}\}$ be a set of $m$ lattice points in $[0,1]^d$ defined $p_i = (\frac{i}{m}, \fra{\alpha_1 i} , \ldots, {\fra{\alpha_{d-1} i}})$.
$P_{\alpha, m}$ is \emph{irrational} with respect to any polytope in $\c{Q}_\beta$ if for all $\beta_i \in \beta$, for all $j \leq d$, and for all $h \leq d-1$, the fraction $\beta_{i,j} / \alpha_h$ is irrational.  
(Note that $\beta_{i,j}$ represents the $j$th element of the vector $\beta_i$.)
Lattices with $\alpha$ irrational (relative to the face normals) generate low discrepancy sets.

\begin{theorem}
\label{thm:tri-d}
Let $Q \in \c{Q}_{\beta^\prime}$ be a fixed convex polytope.  
Let $\beta, \beta^\prime \subset \b{S}^{d-1}$ be sets of $k$ and $k^\prime$ directions, respectively.  
There is an $\eps$-approximation of $(Q, \c{Q}_\beta)$ of size $O((k+k^\prime) \frac{1}{\eps} \log^{d-1} \frac{1}{\eps} \emph{\plle})$.  
\end{theorem}

This $\eps$-approximation is realized by a set of lattice points $P_{\alpha, m} \cap Q$ such that $P_{\alpha, m}$ is irrational with respect to any polytope in $\c{Q}_{\beta \cup \beta^\prime}$.  

\begin{proof}
Consider polytope $t Q_h$ and lattice $P_{\alpha, m}$, where the uniform scaling factor $t$ is treated as an asymptotic quantity.  Skriganov's Theorem 6.1 in \cite{Skr98} claims
$$
\max_{v \in \b{R}^d} \lD(P_{\alpha, m}, t Q_h + v) = O\left(t^{d-1} \rho^{-\theta} + \sum_f S_f(P_{\alpha, m}, \rho) \right)
$$
where 
$$
S_f(P_{\alpha, m}, \rho) = O(\log^{d-1} \rho \log^{1+\tau} \log \rho)
$$
for $\tau >0$, as long as $P_{\alpha, m}$ is irrational with respect to the normal of the face $f$ of $Q_h$ and infinite otherwise, where $\theta \in (0,1)$ and $\rho$ can be arbitrarily large.  
Note that this is a simplified form yielded by invoking Theorem 3.2 and Theorem 4.5 from \cite{Skr98}.  
By setting $\rho^\theta = t^{d-1}$, 
\begin{equation}
\label{eq:Skr-main}
\max_{v \in \b{R}^d} \lD(P_{\alpha,m}, tQ_h + v) = O(h \log^{d-1} t \log^{1+\tau} \log t).
\end{equation}  
Now by noting that as $t$ grows, the number of lattice points in $tQ_h$ grows by a factor of $t^d$, and we can set $t = n^{1/d}$ so (\ref{eq:Skr-main}) implies that $\lD(P_{\alpha, m}, tQ_h) = O(h \log^{d-1} n \log^{1+\tau} \log n)$ for $|P_{\alpha, m}| =m =n$ and $tQ_h \subset [0,1]^d$.  

The discrepancy is a sum over the set of $h$ terms, one for each face $f$, each of which is small as long as $P_{\alpha, m}$ is irrational with respect to $f$'s normal $\beta_f$.  Hence this lattice gives low discrepancy for any polytope in the analogous family $\c{Q}_\beta$ such that $P_{\alpha, m}$ is irrational with respect to $\c{Q}_\beta$.  
Finally we realize that any subset $Q \cap Q_k$ for $Q \in \c{Q}_{\beta^\prime}$ and $Q_k \in \c{Q}_\beta$ is a polytope defined by normals from $\beta^\prime \cup \beta$ and we then refer to $g_{\lD}(\eps, \c{Q}_{\beta \cup \beta^\prime})$ in (\ref{eq:g-LD-eps}) to bound the size of the $\eps$-approximation from the given Lebesgue discrepancy.
\end{proof}

\begin{remark}
Skriganov's result~\cite{Skr98} is proved under the \emph{whole space} model where the lattice is infinite ($tQ_h$ is not confined to $[0,1]^d$), and the relevant error is the difference between the measure of $tQ_h$ versus the cardinality $|tQ_h \cap P_{\alpha, m}|$, where each $p \in P_{\alpha, m}$ represents 1 unit of measure.  Skriganov's main results in this model is summarized in equation (\ref{eq:Skr-main}) and only pertains to a fixed polytope $Q_h$ instead of, more generally, a family of polytopes $\c{Q}_\beta$, as shown in Theorem \ref{thm:tri-d}.
\end{remark}

\mypara{Samples for polygonal terrains.}
Combining the above results and weighted extension of Theorem \ref{thm:thin} implies the following results.  
\begin{theorem}
\label{thm:eps-polyt}
We can create a weighted $\eps$-approximation of size $O((k+k^\prime)\frac{1}{\eps} \cdot \log^{2k} \frac{1}{\eps} \emph{\plle})$ of $(\c{D}, \c{Q}_k)$ in time $O((k+k^\prime) n \frac{1}{\eps^4} \emph{\plog} \frac{1}{\eps})$ for any $d$-dimensional domain $\c{D}$ which can be decomposed into $n$ $d$-dimensional convex $k^\prime$-oriented polytopes.  
\end{theorem}

\begin{proof}
We divide the domain into $n$ $k^\prime$-oriented polytopes and then approximate each polytope $Q_{k^\prime}$ with a point set $P_{\alpha, m} \cap Q_{k^\prime}$ using Theorem \ref{thm:tri-d}.  We observe that the union of these point sets is a weighted $\eps$-approximation of $(\c{D}, \c{Q}_k)$, but is quite large.  Using the weighted extension of Theorem \ref{thm:thin} we can reduce the point sets to the size and in the time stated.  
\end{proof}

This has applications to terrain domains $\c{D}$ defined with a piecewise-linear base $B$ and height function $h : B \to \b{R}$.  We decompose the terrain so that each linear piece of $h$ describes one $3$-dimensional polytope, then apply Theorem \ref{thm:eps-polyt} to get the following result.

\begin{corollary}
\label{cor:eps-ter}
For terrain domain $\c{D}$ with piecewise-linear base $B$ and height function $h : B \to \b{R}$ with $n$ linear pieces, we construct a weighted $\eps$-approximation of $(\c{D},\c{Q}_k)$ of size $O(k\frac{1}{\eps} \log^{4} \frac{1}{\eps}\emph{\plle})$ in time $O(kn \frac{1}{\eps^4} \emph{\plog} \frac{1}{\eps})$.
\end{corollary}

\section{Sampling from Smooth Terrains}
\label{sec:sample-smooth}
We can create an $\eps$-approximation for a smooth domain (one which cannot be decomposed into polytopes) in a three stage process.  The first stage approximates any domain with a set of polytopes.  The second approximates each polytope with a point set.  The third merges all point sets and uses Theorem \ref{thm:thin} to reduce their size.  

This section mainly focuses on the first stage, however, we also offer an improvement for the second stage in a relevant special case.
More formally, we can approximate a non-polygonal domain $\c{D}$ with a set of disjoint polygons $P$ such that $P$ has properties of an $\eps$-approximation.

\begin{lemma}
\label{lem:eps-smooth}
If $|\c{D} \setminus P| \leq \frac{\eps}{2} |\c{D}|$ and $P \subseteq \c{D}$ then 
$\displaystyle{ \; \;
\max_{R \in \c{A}} \left| \frac{|R \cap P|}{|P|} - \frac{|R \cap \c{D}|}{|\c{D}|} \right| \leq \eps.
}$
\end{lemma}

\begin{proof}
No range $R \in \c{A}$ can have $\left| \frac{|R \cap P|}{|P|} - \frac{|R \cap \c{D}|}{|\c{D}|} \right| > \eps$ because if $|\c{D}| \geq |P|$ (w.l.o.g.), then $|R \cap \c{D}| - \frac{|\c{D}|}{|P|} |R \cap P| \leq \eps |\c{D}|$ and $|R \cap P| \frac{|\c{D}|}{|P|} - |R \cap \c{D}| \leq \eps |\c{D}|$.  The first part follows from $\frac{|\c{D}|}{|P|} \geq 1$ and is loose by a factor of 2.  
For the second part we can argue
\begin{eqnarray*}
|R \cap P| \frac{|\c{D}|}{|P|} - |R \cap \c{D}| 
&\leq &
|R \cap P| \frac{1}{1-\frac{\eps}{2}} - |R \cap \c{D}| 
\leq
|R \cap \c{D}|\frac{1}{1-\frac{\eps}{2}} - |R \cap \c{D}|
\\ & = &
\frac{\frac{\eps}{2}}{1-\frac{\eps}{2}} |R \cap \c{D}|
\leq
\eps |R \cap \c{D}|
\leq 
\eps |\c{D}|.
\end{eqnarray*}
\end{proof}

For terrain domains $\c{D}$ defined with a base $B$ and a height function $h : B \to \b{R}$, if $B$ is polygonal we can decompose it into polygonal pieces, otherwise we can approximate it with constant-size polygonal pieces according to Lemma \ref{lem:eps-smooth}.  Then, similarly, if $h$ is polygonal we can approximate the components invoking Corollary \ref{cor:eps-ter}; however, if it is smooth, then we can approximate each piece according to Lemma \ref{lem:eps-smooth}.

Section \ref{sec:vdC} improves on Theorem \ref{thm:tri-d} for the second stage and gives a more efficient way to create an $\eps$-approximation for $(\c{D}, \c{R}_d \times \b{R})$ of a terrain when $B$ is a rectangle and $h$ is linear.  Ranges from the family $\c{R}_d \times \b{R}$ are generalized hyper-cylinders in $d+1$ dimensions where the first $d$ dimensions are described by an axis-parallel rectangle and the $(d+1)$st dimension is unbounded.  
Section \ref{sec:smooth-ter} focuses on the first stage and uses this improvement as a base case in a recursive algorithm (akin to a fair split tree) for creating an $\eps$-approximation for $(\c{D}, \c{R}_d \times \b{R})$ when $B$ is rectangular and $h$ is smooth.

\subsection{Stretching the Van der Corput Set}
\label{sec:vdC}
The Van der Corput set \cite{vdC35} is a point set $P_n = \{p_0, \ldots, p_{n-1}\}$ in the unit square defined for $p_i = (\frac{i}{n}, b(i))$ where $b(i)$ is the bit reversal sequence.  For simplicity we assume $n$ is a power of $2$.  The function $b(i)$ writes $i$ in binary, then reverses the ordering of the bits, and places them after the decimal point to create a number in $[0, 1)$.  For instance for $n=16$, $i = 13 = 1101$ in binary and $b(13) = 0.1011 = \frac{11}{16}$.  Formally, if $i = \sum_{i=0}^{\log n} a_i 2^i$ then $b(i) = \sum_{i=0}^{\log n} \frac{a_i}{2^{i+1}}$.  

Halton \cite{Hal60} showed that the Van der Corput set $P_n$ satisfies $\lD(P_n, \c{R}_2) = O(\log n)$.
We can extend this to approximate any rectangular domain.  For a rectangle $[0,w] \times [0,l]$ (w.l.o.g.) we can use the set $P_{n, w, l}$ where $p_i = (w\cdot \frac{i}{n}, l \cdot b(i))$ and a  version of the Lebesgue discrepancy over a stretched domain is still $O(\log n)$.

We can stretch the Van der Corput set to approximate a rectangle $r = [0,w] \times [0,l]$ with a weighting by an always positive linear height function $h(x,y) = \alpha x + \beta y + \gamma$.  Let $\Delta(w, \alpha, \gamma, i)$ be defined such that the following condition is satisfied
$$\int_0^{\Delta(w, \alpha, \gamma, i)} (\alpha x + \gamma) dx = \frac{i}{n} \int_0^w (\alpha x + \gamma) dx.$$
Note that we can solve for $\Delta$ explicitly and because $h$ is linear it can simultaneously be defined for the $x$ and $y$ direction.  Now define the \emph{stretched Van der Corput set} $S_{n, w, l, h} = \{s_0, \ldots, s_n\}$ for $s_i = (\Delta(w, \alpha, \gamma, i), \Delta(l, \beta, \gamma, b(i)\cdot n))$.

\begin{theorem}
\label{thm:svdC}
For the stretched Van der Corput set $S_{n, w, l, h}$, $\lD(S_{n, w, l, h}, \c{R}_2) = O(\log n)$ over the domain $[0,w] \times [0,l]$ with $h : [0,w] \times [0,l] \to \b{R}^+$ a linear weighting function.
\end{theorem}

The proof follows the proof in Matou\v{s}ek \cite{Mat99} for proving logarithmic discrepancy for the standard Van der Corput set in the unit square.
\begin{proof}
Let a \emph{canonical interval} be of the form 
$\left[ \frac{\Delta(l, \beta, \gamma, k)}{2^q}, \frac{\Delta(l, \beta, \gamma, k+1)}{2^q}\right)$
for integers $q \in [1,n]$ and $k \in [0,2^q)$.  
Let any rectangle $r = [0, a) \times I$ where $I$ is canonical and $a \in (0,1]$ be called a \emph{canonical rectangle}.

\begin{claim} 
\label{clm:can-rect}
For any canonical rectangle $r$, $\lD(S_{n, w, l, h}, r) \leq 1$.
\end{claim}

\begin{proof}
Like in the Van der Corput set, every subinterval of $r$ such that $h(r) = \frac{1}{n}$ has exactly 1 point.  Let $I = \left[ \frac{\Delta(l, \beta, \gamma, k)}{2^q}, \frac{\Delta(l, \beta, \gamma, k+1)}{2^q}\right)$.  Thus each rectangle $r_j = [\Delta(l, \beta, \gamma, \frac{j 2^q}{n}) , \Delta(l, \beta, \gamma, \frac{(j+1) 2^q}{n})) \times I$ contains a single point from $S_{n, w, l, h}$ and $h(r_j) = \frac{1}{n}$, where $h(r) = \int_r h(p) dp$.  

So the only part which generates any discrepancy is the canonical rectangle $r_j$ which contains the segment $a \times I$.  But since $|S_{n, w, l, h} \cap r_j \cap r| \leq 1$ and $h(r_j \cap r) \leq \frac{1}{n}$, the claim is proved.
\end{proof}

Let $\c{C}_d$ be the family of ranges consisting of $d$-dimensional rectangles with the lower left corner at the origin.  Let $C_{(x,y)} \in \c{C}_2$ be the corner rectangle with upper right corner at $(x,y)$.  

\begin{claim}
\label{clm:can-cor}
Any corner rectangle $C_{(x,y)}$ can be expressed as the disjoint union of at most $O(\log n)$ canonical rectangles plus a rectangle $M$ with $|\lD(S_{n, w, l, h}, M)| \leq 1$.  
\end{claim}

\begin{proof}
Starting with the largest canonical rectangle $r_0 = [0, a) \times I$ within $C_{(x,y)}$ such that $I = \left[ \frac{\Delta(l, \beta, \gamma, 0)}{2^q}, \frac{\Delta(l, \beta, \gamma, 1)}{2^q}\right)$ for the smallest value possible of $q$, keep including the next largest disjoint canonical rectangle within $C_{(x,y)}$.  Each consecutive one must increase $q$ by at least $1$.  Thus there can be at most $O(\log n)$ of these.

The left over rectangle $M = [m_x, x] \times [m_y, y]$, must be small enough such that $\int_{0}^w \int_{m_y}^y h(p,q) dq dp < \frac{1}{n}$, thus it can contain at most 1 point and $\lD(S_{n, w, l, h}, M) \leq 1$.
\end{proof}

It follows from Claim \ref{clm:can-rect} and Claim \ref{clm:can-cor} that $\cD(S,\c{C}_2) = O(\log n)$.  We conclude by using the classic result \cite{Mat99} that $\lD(S, \c{C}_2) \leq \lD(S, \c{R}_2) \leq 4 \lD(S, \c{C}_2)$ for any point set $S$.
\end{proof}

This improves on the discrepancy for this problem attained by using Theorem \ref{thm:tri-d} by a factor of $\log \frac{1}{\eps}$.  

\begin{corollary}
\label{cor:vdC}
A stretched Van der Corput set $S_{n, w, l, h}$ forms an $\eps$-approximation of $(\c{D}, \c{R}_2)$ of size $n = O(\frac{1}{\eps} \log \frac{1}{\eps}\emph{\plle})$ for $\c{D}$ defined by a rectangle $[0,w]\times [0,l]$ with a linear height function $h$.  
\end{corollary}

\begin{remark}
This extends to higher dimensions.  A stretched b-ary Van der Corput set~\cite{Mat99} forms an $\eps$-approximation of $(\c{D}, \c{R}_d)$ of size $O(\frac{1}{\eps} \log^{d-1} \frac{1}{\eps}\emph{\plle})$ for $\c{D}$ defined by $\times_{i=1}^d [0,w_i]$ with a linear height function.  Details are omitted.  
\end{remark}


\subsection{Approximating Smooth Terrains}
\label{sec:smooth-ter}
Given a terrain domain $\c{D}$ where $B \subset \b{R}^2$ is rectangular and $h : B \to \b{R}^+$ is a $C^2$-continuous height function we can construct an $\eps$-approximation based on a parameter $\eps$ and properties ${z^-}_{\c{D}}$, $d_{\c{D}}$, and $\lambda_{\c{D}}$.  
Let $z^-_{\c{D}} = \min_{p \in B} h(p)$.  
Let $d_{\c{D}} = \max_{p,q \in B} ||p-q||$ be the diameter of $\c{D}$. 
Let $\lambda_{\c{D}}$ be the largest eigenvalue of $H_h$ where $H_h = \left[\begin{array}{cc} \frac{d^2h}{dx^2} & \frac{d^2 h}{dxdy} \\ \frac{d^2 h}{dydx} & \frac{d^2 h}{dy^2} \end{array} \right]$ is the Hessian of $h$.  

We first create a set of linear functions to approximate $h$ with a recursive algorithm.
If the entire domain cannot be approximated with a single linear function, then we split the domain by its longest direction (either $x$ or $y$ direction) evenly.  This decreases $d_\c{D}$ by a factor of $1/\sqrt{2}$ each time.  We recur on each subset domain.

\begin{lemma}
\label{lem:decomp}
For a domain $\c{D}$ with rectangular base $B \subset \b{R}^2$ with a $C^2$-continuous height function $h : B \to \b{R}$ we can approximate $h$ with $O(\frac{\lambda_{\c{D}} d^2_{\c{D}}}{z^-_{\c{D}} \eps})$ linear pieces $h_\eps$ so that for all $p \in B$ $h_\eps(p) \leq h(p) \leq h_\eps(p) + \eps$.  
\end{lemma}

\begin{proof}
First we appeal to Lemma 4.2 from Agarwal \emph{et. al} \cite{APV06} which says that the error of a first order linear approximation at a distance $d$ is bounded by $\lambda_{\c{D}} d^2$.  Thus we take the tangent at the point in the middle of the range and this linear (first order) approximation has error bounded by $\lambda_{\c{D}} (d_{\c{D}}/2)^2 = \lambda_{\c{D}} d_{\c{D}}^2 /4$.  
The height of the linear approximation is lowered by $\lambda_{\c{D}} d_{\c{D}}^2 /4$ from the tangent point to ensure it describes a subset of $\c{D}$.
Thus, as long as the upper bound on the error $\lambda_{\c{D}} d_{\c{D}}^2/2$ is less than $z_{\c{D}}^- \eps$ then the lemma holds.  
The ratio $\frac{\lambda_{\c{D}} d_{\c{D}}^2}{2 z^-_{\c{D}} \eps}$ is halved every time the domain is split until it is less than $1$.  Thus it has $O(\frac{\lambda_{\c{D}} d_{\c{D}}^2}{z^-_{\c{D}} \eps})$ base cases.
\end{proof}

After running this decomposition scheme so that each linear piece $L$ has error $\eps/2$, we invoke Corollary \ref{cor:vdC} to create an $(\eps/2)$-approximation point set of size $O(\frac{1}{\eps} \log \frac{1}{\eps} \plle)$ for each $(L, \c{R}_2 \times \b{R})$.  The union creates a weighted $\eps$-approximation of $(\c{D}, \c{R}_2 \times \b{R})$, but it is quite large.  We can then reduce the size according to Corollary \ref{cor:compress-Rd} to achieve the following result.

We can improve further upon this approach using a stretched version of the Van der Corput Set and dependent on specific properties of the terrain.  
Consider the case where $B$ is a rectangle with diameter $d_{\c{D}}$ and $h$ is $C^2$ continuous with minimum value $z^-_{\c{D}}$ and where the largest eigenvalue of its Hessian is $\lambda_{\c{D}}$.  For such a terrain $\c{D}$, interesting ranges $\c{R}_2 \times \b{R}$ are generalized cylinders where the first $2$ dimensions are an axis-parallel rectangle and the third dimension is unbounded.  We can state the following result (proved in the full version).

\begin{theorem}
\label{thm:smooth}
For a domain $\c{D}$ with rectangular base $B \subset \b{R}^2$ and with a $C^2$-continuous height function $h:B\to \b{R}$ we can deterministically create a weighted $\eps$-approximation of $(\c{D}, \c{R}_2 \times \b{R})$ of size $O\left(\left(\frac{\lambda_{\c{D}} d_{\c{D}}^2}{z^-_{\c{D}} \eps}\right) \left(\frac{1}{\eps} \log^4 \frac{1}{\eps} \emph{\plle}\right) \right)$.  
We reduce the size to $O(\frac{1}{\eps} \log^{4} \frac{1}{\eps}\emph{\plle})$ in time $O\left(\left(\frac{\lambda_{\c{D}} d_{\c{D}}^2}{z^-_{\c{D}}}\right) \frac{1}{\eps^5} \emph{\plog} \frac{1}{\eps}\right)$.
\end{theorem}

This generalizes in a straightforward way for $B \in \b{R}^d$.  Similar results are possible when $B$ is not rectangular or when $B$ is not even piecewise-linear.  The techniques of Section \ref{sec:sample-poly} are necessary if $Q_k$ is used instead of $\c{R}_2$, and are slower by a factor $O(\frac{1}{\eps})$.


\section{Applications}
\label{sec:app}
Creating smaller $\eps$-approximations improves several existing algorithms.  

\subsection{Terrain Analysis}
\label{sec:terrain}
After creating an $\eps$-approximation of a terrain we are able to approximately answer questions about the integral over certain ranges.  For instance, a terrain can model the height of a forest.  A foresting company may deem a region ready to harvest if the average tree height is above some threshold.  Computing the integral on the $\eps$-approximation will be much faster than integrating over the full terrain model.  

More interesting analysis can be done by comparing two terrains.  These can represent the forest height and the ground height or the elevation of sand dunes at two snapshots or the distribution of a population and a distribution of a trait of that population.  
Let $T_1$ and $T_2$ be two terrains defined by piecewise-linear height functions $h_1$ and $h_2$, respectively, over a subset of $\b{R}^2$.  
The height $h = h_1 - h_2$ may be negative in some situations.  This can be handled by dividing it into two disjoint terrains, where one is the positive parts of $h$ and the other is the negative parts.  Each triangle can be split by the $h = 0$ plane at most once, so this does not asymptotically change the number of piecewise-linear sections.  

Once an $\eps$-approximation has been created for the positive and negative terrain, the algorithms of Agarwal \emph{et. al.} \cite{APV06} can be used to find the rectangle with the largest positive integral.  For $n$ points this takes $O(n^2 \log n)$ time.  The same can be done for finding the rectangular range with the most negative integral.  The range returned indicates the region of largest difference between the two terrains.  The runtime is dominated by the time to create the $\eps$-approximation in Corollary \ref{cor:eps-ter}.

\subsection{Biosurveillance}
\label{sec:sss}
Given two points set representing measured data $M$ and representing baseline data $B$, anomaly detection algorithms find the region where $M$ is most different from $B$.  The measure of difference and limits on which regions to search can vary significantly~\cite{Kul97,AMPVZ06,PT04}.  One well-formed and statistically justified definition of the problem defines the region $R$ from a class of regions $\c{A}$ that maximizes a discrepancy function based on the notion of spatial scan statistics~\cite{Kul97,APV06}.  Where $m_R= |R \cap M| / |M|$ and $b_R= |R \cap B|/|B|$ represent the percentage of the baseline and measured distributions in a range $R$, respectively, then the Poisson scan statistic can be interpreted as the Poisson discrepancy function
$d_P(m_R, b_R) = m_R \ln \frac{m_R}{b_R} + (1-m_R) \ln \frac{1-m_R}{1-b_R}$.
This has important applications in biosurveillance~\cite{Kul97} where $B$ is treated as a population and $M$ is a subset which has a disease (or other condition) and the goal is to detect possible regions of outbreaks of the disease as opposed to random occurrences.  
We say a \emph{linear discrepancy function} is of the form $d_l(m_R, b_R) = \alpha m_R + \beta b_R + \gamma$ for constants $\alpha$, $\beta$, and $\gamma$.  The Poisson discrepancy function can be approximated up to an additive $\eps$ factor with $O(\frac{1}{\eps} \log^2 n)$ linear discrepancy functions~\cite{APV06}.  The range $R \in \c{R}_2$ which maximizes a linear discrepancy function can be found in $O(n^2 \log n)$ time and the $R \in \c{R}_2$ which maximizes any discrepancy can be found in $O(n^4)$ time where $|B| + |M| = n$.  

Agarwal \emph{et. al.} \cite{AMPVZ06} note that a random sample of size $O(\frac{1}{\eps^2} \log \frac{1}{\eps \delta})$ will create an $\eps$-approximation with probability $1-\delta$.  This can be improved using Corollary \ref{cor:compress-Rd}.  We then conclude:

\begin{theorem}
Let $|M \cup B| = n$.  
A range $R \in \c{R}_2$ such that $|d_P(m_R, b_R) - \max_{r \in \c{R}_2} d_P(m_r, b_r)| \leq \eps$ can be deterministically found in $O(n \frac{1}{\eps^3} \emph{\plle} + \frac{1}{\eps^4} \emph{\plle})$ time.

A range $R \in \c{R}_2$ such that $|d_P(m_R, b_R) - \max_{r \in \c{R}_2} d_P(m_r, b_r)| \leq \eps + \delta$ can be deterministically found in $O(n \frac{1}{\eps^3} \emph{\plle} + \frac{1}{\delta}\frac{1}{\eps^2} \emph{\plle})$ time.
\label{thm:stat-disc}
\end{theorem}

This can be generalized to when $M$ and $B$ are terrain domains.  This case arises, for example, when each point is replaced with a probability distribution.

\mypara{Generating Terrains with Kernel Functions.}
\label{sec:kernel}
A drawback of the above approach to finding maximum discrepancy rectangles is that it places the boundaries of rectangles in some arbitrary place between samples.  This stems from the representation of each sample as a fixed point.  In reality its location is probably given with some error, so a more appropriate model would be to replace each point with a kernel density function.  Probably, the most logical kernel function would be a Gaussian kernel, however, this is continuous and its tails extend to infinity.  The base domain can be bounded to some polygon $B$ so that the integral under the kernel function outside of $B$ is less than $\eps$ times the entire integral and then Theorem \ref{thm:smooth} can be applied.  
(See Appendix \ref{app:normal} for details.)
Alternatively, we can replace each point with a constant complexity polygonal kernel, like a pyramid.  Now we can ask questions about spatial scan statistics for a measured $T_M$ and a baseline $T_B$ terrain.  

For simple ranges such as $\c{H}_{||}$ (axis-parallel halfspaces) and $\c{S}_{||}$ (axis-parallel slabs) finding the maximal discrepancy ranges on terrains reduces to finding maximal discrepancy intervals on points sets in $\b{R}^1$.

\begin{theorem}
\label{thm:maxterrain}
For a terrain $T$ defined by a piecewise-linear height function $h$ with $n$ vertices 
$$ \arg \max_{R \in \c{H}_{||}} \int_R h(p) \; dp \;\;\; and \;\;\; \arg \max_{R \in \c{S}_{||}} \int_R h(p) \; dp$$
can be found in $O(n)$ time.
\end{theorem}

\begin{proof} 
\emph{(sketch)}
Project all terrains onto the axis perpendicular to halfplanes (or slabs).  Integrate between points where the projected terrain crosses $0$.  Treat these intervals as weighted points and use techniques from Agarwal \emph{et. al.} \cite{APV06}.  
The full proof is given in Appendix \ref{app:comb-terrain}.  
\end{proof}

However for $\c{R}_2$, this becomes considerably more difficult.  Under a certain model of computation where a set of 4 quadratic equations of 4 variables can be solved in constant time, the maximal discrepancy rectangle can be found in $O(n^4)$ time.  However, such a set of equations would require a numerical solver, and would thus be solved approximately.  But using Theorem \ref{thm:stat-disc} we can answer the question within $\eps n$ in $O(n \frac{1}{\eps^3} \plog \frac{1}{\eps} + \frac{1}{\eps^4} \plog \frac{1}{\eps})$ time for a terrain with $O(n)$ vertices.  

Alternatively, we can create an $\eps$-approximation for a single kernel, and then replace each point in $M$ and $B$ with that $\eps$-approximation.  Appendix \ref{app:normal} describes, for a Gaussian function $\varphi$, how to create an $\eps$-approximation for $(\varphi, \c{R}_2)$ of size $O(\frac{1}{\eps} \plle)$ in time $O(\frac{1}{\eps^7} \plle)$.  For standard kernels, such as Guassians, we can assume such $\eps$-approximations may be precomputed.
We can then apply Corollary \ref{cor:compress-Qk}, before resorting to techniques from Agarwal \etal~\cite{APV06}.

\begin{theorem}
Let $|M \cup B| = n$, where $M$ and $B$ are point sets describing the centers of Gaussian kernels with fixed variance.  For a range $R \in \c{R}_2$, let $m_R = \frac{\int_{x \in R} M(x)}{\int_{x \in \b{R}^2} M(x)}$ and let $b_R = \frac{\int_{x \in R} M(x)}{\int_{x \in \b{R}^2} M(x)}$.  

A range $R \in \c{R}_2$ such that $|d_P(m_R, b_R) - \max_{r \in \c{R}_2} d_P(m_r, b_r)| \leq \eps$ can be deterministically found in $O(n \frac{1}{\eps^4} \emph{\plle})$ time.
%
\end{theorem}

\subsection{Cuts in Sensor Networks}
\label{sec:sn-cut}
Sensor networks geometrically can be thought of as a set of $n$ points in a domain $\c{D}$.  These points (or nodes) need to communicate information about their environment, but have limited power.  Shrivastava \emph{et. al.} \cite{SST05} investigates the detection of large disruptions to the domain that affect at least $\eps n$ nodes.  They want to detect these significant events but with few false positives.  In particular, they do not want to report an event unless it affects at least $\frac{\eps}{2} n$ nodes.  

We say $P \subseteq \c{D}$ is an \emph{$\eps$-sentinel} of $(\c{D},\c{A})$ if for all $R \in \c{A}$ 
\begin{itemize}
\item if $|R \cap \c{D}| \geq \eps |\c{D}|$ then $|R \cap P| \geq \eps \frac{3}{4} |P|$, and  
\item if $|R \cap P| \geq \eps \frac{3}{4} |P|$ then $|R \cap \c{D}| \geq \frac{\eps |\c{D}|}{2}$.
\end{itemize}
Shrivastava \emph{et. al.} \cite{SST05} construct $\eps$-sentinels for half spaces of size $O(\frac{1}{\eps})$ and in expected time $O(\frac{n}{\eps} \log n)$.
They note that an $\eps/4$-approximation can be used as an $\eps$-sentinel, but that the standard upper bound for $\eps$-approximations~\cite{VC71} requires roughly $O(\frac{1}{\eps^2} \log \frac{1}{\eps})$ points which is often impractical.  They pose the question:  \emph{For what other classes of ranges can an $\eps$-sentinel be found of size $O(\frac{1}{\eps} \emph{\plog} \frac{1}{\eps})$?}

Followup work by Gandhi \emph{et. al.} \cite{GSW07} construct $\eps$-sentinels for any $\c{A}$ with bounded VC-dimension $v$ (such as disks or ellipses) of size $O(\frac{1}{\eps} \log \frac{1}{\eps})$ and in time $O(n \frac{1}{\eps^{2v}} \log^v \frac{1}{\eps})$.  

As an alternative to this approach, by invoking Corollary \ref{cor:compress-Qk} we show that we can construct a small $\eps$-sentinel for $\c{Q}_k$.

\begin{theorem}
For a discrete point set $\c{D}$ of size $n$, we can compute $\eps$-sentinels for $(\c{D}, \c{Q}_k)$ of size $O(\frac{1}{\eps} \log^{2k} \frac{1}{\eps} \emph{\plle})$ in time $O(n \frac{1}{\eps^3} \plle)$.
\end{theorem}

In fact, if we can choose where we place our nodes we can create an $\eps$-sentinel of size $O(\frac{1}{\eps} \plog \frac{1}{\eps})$ to monitor some domain $\c{D}$.  We can invoke Theorem \ref{thm:tri-d} or Theorem \ref{thm:smooth}, depending on the nature of $\c{D}$.

Additionally, by modifying the techniques of this paper, we can create $O(n \eps / \log^{2k} \frac{1}{\eps})$ disjoint sets of $\eps$-sentinels.  At every \textsc{Halve} step of Algorithm \ref{alg:trim} we make a choice of which points to discard.  By branching off with the other set into a disjoint $\eps$-approximation, we can place each point into a disjoint $\eps$-sentinel of size $O(\frac{1}{\eps} \log^{2k} \plle)$.  
Since the \textsc{Halve} step now needs to be called $O(n \eps / \log^{2k} \frac{1}{\eps})$ times on each of the $O(\log ({n \eps}))$ levels, this takes $O(n \frac{1}{\eps^3} \log (n \eps) \plle)$ time.  

\begin{theorem}
For a discrete point set $\c{D}$ of size $n$, we can create $O(n \eps / \log^{2k} \frac{1}{\eps})$ disjoint sets of $\eps$-sentinels in $O(n \frac{1}{\eps^3} \log(n\eps) \emph{\plle})$ total time.  
\end{theorem}

The advantage of this approach is that the nodes can alternate which sensors are activated, thus conserving power.  If instead a single node is used in multiple $\eps$-sentinels it will more quickly use up its battery supply, and when its batter runs out, the $\eps$-sentinels using that node can no longer make the appropriate guarantees.

\section*{Acknowledgements}
I would like to thank Pankaj Agarwal for many helpful discussions including finding a bug in an earlier version of the proof of Lemma \ref{thm:cd-pol}, Shashidhara Ganjugunte, Hai Yu, Yuriy Mileyko, and Esther Ezra for a careful proofreading, Jirka Matou\v{s}ek for useful pointers, Subhash Suri for posing a related problem, and Don Rose for discussions on improving the Beck-Fiala Theorem.

\bibliographystyle{plain}
\bibliography{discbib}

\appendix{\LARGE{\textbf{APPENDIX}}}

\section{Combinatorial Algorithms on Terrains}
\label{app:comb-terrain}

\subsection{Half spaces, intervals, and slabs}
\label{app:his}
Let $h : \b{R} \to \b{R}$ be a piecewise-linear height function over a one-dimensional domain with a possibly negative range.  Each range in $\c{H}_{||}$ is defined by a single point in the domain.  

\begin{lemma}
For continuous $h : \b{R} \to \b{R}$ the 
$$\arg \max_{R \in \c{H}_{||}} \int_R h(p) \; dp$$
is defined by an endpoint $r$ such that $h(r) = 0$.  
\end{lemma}
\begin{proof}
If the end point $r$ moved so the size of $R$ is increased and $h(r) > 0$ then the integral would increase, so $h(r)$ must be non positive.  If the end point $r$ is moved so the size of $R$ is decreased and $h(r) < 0$ then integral would also increase, so $h(r)$ must be non negative.  
\end{proof}

This proof extends trivially to axis-parallel slabs $\c{S}_{||}$ (which can be thought of as intervals) as well.  
\begin{lemma}
For continuous $h : \b{R} \to \b{R}$, the 
$$\arg \max_{R \in \c{S}_{||}} \int_R h(p) \; dp$$
is defined by two endpoints $r_l$ and $r_r$ such that $h(r_l) = 0$ and $h(r_r) = 0$.  
\end{lemma}

Let $h$ have $n$ vertices.  For both $\c{H}_{||}$ and $\c{S}_{||}$, the optimal range can be found in $O(n)$ time.  For $\c{H}_{||}$, simply sweep the space from left to right keeping track of the integral of the height function.  When the height function has a point $r$ such that $h(r) =0$, compare the integral versus the maximum so far.  

For $\c{S}_{||}$, we reduce this to a one-dimensional point set problem.  First sweep the space and calculate the integral in between every consecutive pair of points $r_1$ and $r_2$ such that $h(r_1) = 0 = h(r_2)$ and there is no point $r_3$ such that $h(r_3) = 0$ and $r_1 < r_3 < r_2$.  Treat each of these intervals as a point with weight set according to its value.  Now run the algorithm from Agarwal \etal~\cite{APV06} for linear discrepancy of red and blue points where the positive intervals have a red weight equal to the integral and the negative intervals have a blue weight equal to the negative of the integral.  

\begin{theorem}
For continuous, piecewise-linear $h : \b{R} \to \b{R}$ with $n$ vertices 
$$\arg \max_{R \in \c{H}_{||}} \int_R h(p) \; dp  \; \; \; \text{ and } \; \; \; \arg \max_{R \in \c{S}_{||}} \int_R h(p) \; dp$$ 
can be calculated in $O(n)$ time.
\end{theorem}

This result extends trivially to a higher dimensional domains as long as the families of ranges are no more complicated.  

\begin{theorem} \textbf{\emph{[Theorem \ref{thm:maxterrain}]}}
For continuous, peicewise-linear $h : \b{R}^d \to \b{R}$ with $n$ vertices, 
$$\arg \max_{R \in \c{H}_{||}} \int_R h(p) \; dp  \; \; \; \text{ and } \; \; \; \arg \max_{R \in \c{S}_{||}} \int_R h(p) \; dp$$
can be calculated in $O(n)$ time.
\end{theorem}
\begin{proof}
The sweep step of the algorithms described above are performed in the same way, only now the integral of up to $O(n)$ cubic functions must be calculated.  However, this integral can be stored implicitly as a single linear function and can be updated in constant time every time a new vertex is reached.  
\end{proof}

Finally, we not a slightly surprising theorem about the difference between two terrains.
\begin{theorem}
Let $M : \b{R}^d \to \b{R}$ and $B : \b{R}^d \to \b{R}$ be piecewise-linear functions with $n$ and $m$ vertices respectively.  
$$
\arg \max_{R \in \c{H}_{||}} \int_R M(p)-B(p) \; dp  \; \; \; \text{ and } \; \; \; \arg \max_{R \in \c{S}_{||}} \int_R M(p)-B(p) \; dp
$$
can be calculated in $O(n+m)$ time.
\end{theorem}
Naively, this could be calculated in $O(nm)$ time by counting the vertices on the terrain $h(p) = M(p)-B(p)$.  But we can do better.
\begin{proof}
Although there are more than $n+m$ vertices in $h(p) = M(p)-B(p)$, the equations describing the height functions only change when a vertex of one of the original functions is encountered.  Thus there are only $O(n+m)$ linear functions which might cross $0$.  These can be calculated by projecting $M$ and $B$ independently to the axis of $\c{H}_{||}$ or $\c{S}_{||}$ and then taking their sum between each consecutive vertex of either function.
\end{proof}

\subsection{Rectangles}
Although the work by Agarwal \etal \cite{APV06} extends the one-dimensional case for point sets to a $O(n^2 \log n)$ algorithm for rectangles, when the data is given as picewise-linear terrains the direct extension does not go through.  However, a simple $O(n^4)$ time algorithm, under a certain model, does work.  Following algorithm \textbf{Exact} from Agarwal \etal \cite{AMPVZ06}, we make four nested sweeps over the data.  The first two bound the $x$ coordinates and the second two bound the $y$ coordinates.  The inner most sweep keeps a running total of the integral in the range.  However, unlike \textbf{Exact} each sweep does not give an exact bound for each coordinate, rather it just restricts its position between two vertices.  The optimal position is dependent on all four positions, and needs to be determined by solving a system of four quadratic equations.  This system seems to in general have no closed form solution (see the next subsection) and needs to be done via a numerical solver.  However, these equations can be updated in constant time in between states of each sweep, so under the model that the numerical solver takes $O(1)$ time, this algorithm runs in $O(n^4)$ time.  

For the full correctness of the algorithm, there is actually one more step required.  Given that each side of the rectangle is bounded between two vertices, the set of four equations is dependent on which face of the terrain that the corner of the rectangle lies in.  It turns out that each possible corner placement can be handled individually without affecting the asymptotics.  The $n$ vertices impose an $n \times n$ grid on the domain, yielding $O(n^2)$ grid cells in which a corner may lie.  Because the terrain is a planar map, there are $O(n)$ edges as well, and each edge can cross at most $O(n)$ grid cells.  Since no two edges can cross, this generates at most $O(n^2)$ new regions inside of all $O(n^2)$ grid cells.  Since each rectangle is determined by the placement of two opposite corners the total complexity is not affected and is still $O(n^4)$.  We summarize in the following lemma.

\begin{lemma}
Consider a model where a system of $4$ quadratic equations can be solved in $O(1)$ time.  Then let $h : \b{R}^2 \to \b{R}$ be a piecewise-linear function with $n$ vertices.
$$
\arg \max_{R \in \c{R}_2} \int_R h(p) \; dp
$$
can be solved in $O(n^4)$ time. 
\end{lemma}

\subsection{Equations}
\label{sec:eq}
For a piecewise-linear terrain $h : \b{R}^2 \to \b{R}$ we wish to evaluate $D(h,R) = \int_R h(p) \; dp$ where $R$ is some rectangle over the domain of $h$.  Within $R$, the value of $h$ is described by a set of triangles $T_R = \{t_1, \ldots, t_k\}$.  Let $R$ be described by its four boundaries.  Let $x_1$ and $x_2$ describe the left and right boundaries, respectively, and let $y_1$ and $y_2$ describe the top and bottom boundaries, respectively.  Now
$$D(h,R) = \int_{x_1}^{x_2} \int_{y_1}^{y_2} h(x,y) \; dy dx.$$

Assume that we have computed the integral from $x_1$ up to $x(v)$ the $x$-coordinate of a vertex $v$ in the piecewise-linear terrain.  To extend the integral up to the $x(u)$ where $u$ is the next vertex to the right.  We need to consider all of the triangles that exist between $x(v)$ and $x(u)$.  Label this set $\{t_1, \ldots, t_k\}$ where $t_i$ is below $t_j$ in the $y$-coordinate sense for $i<j$.  Note that no triangle can begin or end in this range and this order must be preserved.  We also consider the intersection between the edge of the triangulation and $R$ a vertex.  Let the slope within a triangle $t_i$ be described
\begin{equation}
\label{eq:height-t}
h_i(x,y) = \alpha_i x + \beta_i y + \gamma_i
\end{equation}
and describe the edge of the triangulation that divides $t_i$ and $t_{i+1}$ as
\begin{equation}
\label{eq:line}
l_{i} = \omega_i x + \kappa_i.
\end{equation}
Now the integral from $x(v)$ to $x(u)$ is described
\begin{equation*}
\int_{x(v)}^{x(u)} \int_{y_1}^{y_2} h(x,y) dy dx 
\end{equation*}
\small
\begin{eqnarray*}
 & = & 
\int_{x(v)}^{x(u)} \left[ \int_{y_1}^{l_1(x)} h_1(x,y) dy + \sum_{i=2}^{k-1} \int_{l_{i-1}(x)}^{l_i(x)} h_i(x,y) dy+ \int_{l_{k-1}(x)}^{y_2} h_k(x,y) dy \right] dx 
\\ & = &
\int_{x(v)}^{x(u)} \left[ \begin{array}{l}
\int_{y_1}^{\omega_1 x + \kappa_1} (\alpha_1 x + \beta_1 y + \gamma_1) dy + 
\\ \sum_{i=2}^{k-1} \int_{\omega_{i-1} x + \kappa_{i-1}}^{\omega_i x + \kappa_i} (\alpha_i x + \beta_i y + \gamma_i) dy + 
\\ \int_{\omega_{k-1} x + \kappa_{k-1}}^{y_2} (\alpha_k x + \beta_k y + \gamma_k) dy
\end{array}
\right] dx
\\ & = &
\int_{x(v)}^{x(u)} \left[ 
\begin{array}{l}
\alpha_1 x y + \frac{1}{2} \beta_1 y^2 + \gamma_1 y \mid_{y = y_1}^{\omega_1 x+ \kappa_1} + \\
\sum_{i=2}^{k-1} \alpha_i x y + \frac{1}{2} \beta_1 y^2 + \gamma_i y \mid_{y = \omega_{i-1} x + \kappa_{i-1}}^{\omega_1 x + \kappa_1} + \\
\alpha_{k-1} x y + \frac{1}{2} \beta_{k-1} y^2 + \gamma_{k-1} y \mid_{y = \omega_{k-1} x + \kappa_{k-1}}^{y_2}
\end{array}
\right] dx
\\ & = & 
\int_{x(v)}^{x(u)} \left[
\begin{array}{l}
(\alpha_1 x y_1 + \frac{1}{2}\beta_1 (y_1)^2 + \gamma_1 y_1) - (\alpha_1 x (\omega_1 x + \kappa_1) + \frac{1}{2}\beta_1 (\omega_1 x + \kappa_1)^2 + \gamma_1 (\omega_1 x + \kappa_1)) + \\
\sum_{i=2}^{k-1} \left(\begin{array}{l} \left(
\alpha_i x (\omega_{i-1} x + \kappa_{i-1}) + \frac{1}{2} \beta_i (\omega_{i-1} x + \kappa_{i-1})^2 + \gamma_i (\omega_{i-1} x + \kappa_{i-1})\right) - \\ \left(\alpha_i x (\omega_i x + \kappa_i) + \frac{1}{2} \beta_i (\omega_i x + \kappa_i)^2 + \gamma_i (\omega_i x + \kappa_i) \right)
\end{array}\right) + \\
(\alpha_k x (\omega_{k-1} x + \kappa_{k-1})+ \frac{1}{2}\beta_k (\omega_{k-1} x + \kappa_{k-1})^2 + \gamma_k (\omega_{k-1} x + \kappa_{k-1})) - (\alpha_k x y_2 + \frac{1}{2} \beta_k (y_2)^2 \gamma_k y_2)
\end{array}
\right] dx
\\ & = & 
\int_{x(v)}^{x(u)} \left[
\begin{array}{l}
(\alpha_1xy_1 + \frac{1}{2} \beta_1 y_1^2 + \gamma_1 y_1) + \\
\frac{1}{2} \sum_{i=1}^{k-1} (\alpha_{i+1}-\alpha_i)x + (\beta_{i+1} - \beta_i) (\omega_i^2 x^2 + 2 \kappa_i \omega_i x + \kappa_i^2) + (\gamma_{i+1} - \gamma_i) - \\
(\alpha_kxy_2 + \frac{1}{2} \beta_k y_2^2 +  \gamma_k y_2)
\end{array}
\right] dx
\\ & = & 
\begin{array}{l}
\frac{1}{2}\alpha_1 y_1 x^2 + (\frac{1}{2} \beta_1 y_1^2 + \gamma_1 y_1) x +  \\
\frac{1}{2} \sum_{i=1}^{k-1} \frac{1}{2}(\alpha_{i-1} - \alpha_i + 2 \kappa_i \omega_i (\beta_{i+1} - \beta_i)) x^2 + \frac{1}{3}(\omega_i^2 (\beta_{i+1} - \beta_i)) x^3 + (\gamma_{i+1} - \gamma_i + \kappa_i (\beta_{i+1} - \beta_i)) x - \\
\frac{1}{2}\alpha_k y_2 x^2 - (\frac{1}{2} \beta_k y_2^2 + \gamma_k y_2) x
\end{array}
\mid_{x(v)}^{x(u)}
\\ & = &
\begin{array}{l}
\frac{1}{2}(\alpha_1 y_1 - \alpha_k y_2) (x(u)^2 - x(v)^2) + (\frac{1}{2} \beta_1 y_1^2 - \frac{1}{2} \beta_k y_2^2 + \gamma_1 y_1 - \gamma_k y_2) ( x(u) - x(v)) + \\
\frac{1}{2} \sum_{i=1}^{k-1} \begin{array}{l} 
\frac{1}{2}(\alpha_{i-1} - \alpha_i + 2 \kappa_i \omega_i (\beta_{i+1} - \beta_i)) (x(u)^2 - x(v)^2) + 
\\  \frac{1}{3}(\omega_i^2 (\beta_{i+1} - \beta_i)) (x(u)^3 - x(v)^3) + 
\\ (\gamma_{i+1} - \gamma_i + \kappa_i (\beta_{i+1} - \beta_i)) (x(u) - x(v)) 
\end{array}
\end{array}
\end{eqnarray*}
\normalsize

The point of all of this computation is that $\cD(h, [x_1, x_2] \times [y_1, y_2]) =  \int_{x_1}^{x_2} \int_{y_1}^{y_2} h(x,y) \; dy dx$ is a third order polynomial in $x_2$, the $x$-position of the right endpoint, given that the set of triangles defining $h$ is fixed.  Thus to solve
$$\arg \min_{x_2} \cD(h, [x_1, x_2] \times [y_1, y_2])$$
requires finding where $\frac{\partial}{\partial x_2} \cD(h, [x_1, x_2] \times [y_1, y_2]) = 0$.  If it is outside the range $[x(v), x(u)]$, the two vertices bounding $x_2$, then it is minimized at one of the two end points.  It should be noted that by symmetry, the minimum of $x_1$ and $x_2$ are independent, given $y_1$ and $y_2$, but that both are dependent on $y_1$ and $y_2$.  Also, by symmetry, the previous statements can swap every $x_1$ and $x_2$ with $y_1$ and $y_2$, respectively.  
 
Solving for $\arg \min_{[x_1, x_2] \times [y_1, y_2]} \cD(h, [x_1, x_2] \times [y_1, y_2])$ requires solving $4$ quadratic equations of $4$ variables.  This system is of the form

\begin{equation*}
\left[ \begin{array}{c}   0 \\ 0 \\ 0 \\ 0  \end{array} \right] = 
\left[ \begin{array}{cccc}  
0 & 0 & \alpha_{y_1} & \alpha_{y_2} \\ 
0 & 0 & \beta_{y_1} & \beta_{y_2} \\
\alpha_{x_1} & \alpha_{x_2} & 0 & 0 \\
\beta_{x_1} & \beta_{x_2} & 0 & 0 \\
\end{array} \right] 
\left[ \begin{array}{c} x_1 \\ x_2 \\ y_1 \\ y_2
\end{array} \right] + 
\left[ \begin{array}{cccc}  
0 & 0 & \omega_{y_1} & \omega_{y_2} \\ 
0 & 0 & \kappa_{y_1} & \kappa_{y_2} \\
\omega_{x_1} & \omega_{x_2} & 0 & 0 \\
\kappa_{x_1} & \kappa_{x_2} & 0 & 0 \\
\end{array} \right] 
\left[ \begin{array}{c} x_1^2 \\ x_2^2 \\ y_1^2 \\ y_2^2
\end{array} \right] + 
\left[ \begin{array}{c}   \gamma_{x_1} \\ \gamma_{x_2} \\ \gamma_{y_1} \\ \gamma_{y_2}  \end{array} \right] 
\end{equation*}
which in general has no closed-form solution.  




\section{$\eps$-approximations for a Normal Distribution}
\label{app:normal}

A \emph{normal distribution}, often referred to as a Gaussian distribution, is a family of continuous distributions often used to model error.  The \emph{central limit theorem} highlights its power by showing that the sum of independent and identically distributed distributions with bounded variance converge to a normal distribution as the set size grows.  A normal distribution $\nor(m,\sigma^2)$ is defined by two parameters, a \emph{mean} $m$ marking the center of the distribution and a \emph{variance} $\sigma^2$ scaling the spread of the distribution.  Specifically, we can analyze a random variable $X$ distributed according to a normal distribution, denoted $X \sim \nor(m,\sigma^2)$.  We then say that 
$$
\varphi_{m, \sigma^2}(x) = \frac{1}{\sigma \sqrt{2 \pi}} e^{-\frac{(x- m)^2}{2 \sigma^2}}
$$
describes the probability that a point $X$ drawn randomly from a normal distribution $\nor(m, \sigma^2)$ is located at $x \in \b{R}$.  
Since it is a distribution, then $\int_{x \in \b{R}} \varphi_{m, \sigma^2}(x) = 1$.  
A \emph{standard normal distribution} $\nor(0,1)$ has mean $0$ and variance $1$.  As the variance changes, the normal distribution is stretched symmetrically and proportional to $\sigma$ so the integral is still $1$.  The inflection points of the curve describing the height of the distribution are at the points $m - \sigma$ and $m + \sigma$.  

A \emph{multivariate normal distribution} is a higher dimensional extension to the normal distribution.  A $d$-dimensional random variable $X = [X_1, \ldots, X_d]^T$ is drawn from a multivariate normal distribution if for every linear combination $Y = a_1 X_1 + \ldots + a_d X_d$ (defined by any set of $d$ scalar values $a_i$) is normally distributed.  Thus for a $d$-dimensional random variable $X$ defined over the domain $\b{R}^d$, any projection of $X$ to a one dimensional subspace of $\b{R}^d$ is normally distributed.  

We now discuss how to create $\eps$-approximations for $(\c{D}, \c{R}_2 \times \b{R})$ where $\c{D}$ is a multivariate normal distribution with domain $\b{R}^2$.  Extensions to higher dimensions will follow easily. 
We primarily follow the techniques outlined in Section \ref{sec:smooth-ter} for a smooth terrain with properties $z^-_{\c{D}}$, $d_{\c{D}}$, and $\lambda_{\c{D}}$.  What remains is to approximate $\c{D}$ with another domain $\c{D}^\prime$ such that $\c{D}$ has better bounds on the quantities $z^-_{\c{D}^\prime}$ and $d_{\c{D}^\prime}$.  The approximation will obey Lemma \ref{lem:eps-smooth}, replacing $P$ with $\c{D}^\prime$ by just truncating the domain of $\c{D}$.  

The cumulative distribution function $\Phi_{m, \sigma^2}(x)$ for a normal distribution $\varphi_{m,\sigma^2}$ describes the probability that a random variable $X \sim \nor(m,\sigma^2)$ is less than or equal to $x$.  We can write
$$
\Phi_{m, \sigma^2}(x) 
= 
\int_{-\infty}^x \varphi_{m,\sigma^2}(t) \; dt
=
\frac{1}{\sigma \sqrt{2 \pi}} \int_{-\infty}^x e^{-\frac{(t-m)^2}{2 \sigma^2}} \; dt
=
\frac{1}{2} \left(1 + \erf \left( \frac{x-m}{\sigma \sqrt{2}} \right) \right),
$$
where 
$\erf(x) = \frac{2}{\sqrt{\pi}} \int_{0}^x e^{-t^2} \; dt$.  
W.l.o.g. we can set $m = 0$.
We want to find the value of $x$ such that $\Phi_{0, \sigma^2}(x) \geq 1 - \eps/4$ so that if we truncate the domain of $\varphi_{0, \sigma^2}(x)$ which is being approximated, the result will still be within $\eps/2$ of the original domain.  We can bound $\erf(x)$ with the following inequality which is very close to equality as $x$ becomes large:
$$
1-\erf(x) \leq \frac{e^{-x^2}}{x \sqrt{\pi}}.  
$$
For a $\alpha \leq 1 / (e \sqrt{\pi})$ then 
$$
1 - \erf(x) \leq \alpha 
\;\; \textrm{ when } \;\; 
x \geq \sqrt{- \ln (\alpha \sqrt{\pi})}.
$$
We say that the tail of $\Phi_{0, \sigma^2}$ is sufficiently small when $1-\Phi_{0, \sigma^2} =\frac{1}{2} - \frac{1}{2} \erf(x/(\sigma \sqrt{2})) \leq \eps/4$.  Thus letting $\alpha = \eps/4$, this is satisfied when 
$$
x \geq \sigma \sqrt{2 \ln (1 / (\eps \sqrt{\pi/2}))}.
$$
Thus 
$$
\int_{m - \sigma \sqrt{2 \ln (1 / (\eps \sqrt{\pi/2})}}^{m+\sigma \sqrt{2 \ln (1 / (\eps \sqrt{\pi/2}))}} \varphi_{m, \sigma^2}(x) \; dx \geq 1-\eps/2
$$
and bounding the domain of the normal distribution $\varphi_{m, \sigma^2}$ to 
$$\left[m - \sigma \sqrt{2 \ln (1/ (\eps \sqrt{\pi/2}))}, m + \sigma \sqrt{2 \ln (1 / (\eps \sqrt{\pi/2}))}\right]$$ 
will approximate the distribution within $\eps/2$.  

For a multivariate normal distribution, we truncate in the directions of the $x$- and $y$-axis according to the variance in each direction.  
Letting $\c{D}^\prime$ be the normal distribution with this truncated domain, then the diameter of the space is $d_{\c{D}^\prime} = O(\sigma_{\max} \sqrt{\log (1 / \eps)})$, where $\sigma_{\max}$ is the maximum variance over all directions.  ($\sigma_{\min}$ is the minimum variance.)  In $d$-dimensions, the diameter is $d_{\c{D}^\prime} = O(\sigma_{\max} \sqrt{d \log (1/\eps)})$.  

The lower bound $z^-_{\c{D}^\prime}$ is now on the boundary of the truncation.  In one dimension, the value at the boundary is
$$
\varphi_{0, \sigma^2}\left(\sigma \sqrt{2 \ln(1/ (\eps \sqrt{\pi/2}))}\right) 
= 
\frac{1}{\sigma \sqrt{2\pi}} e^{-\left(\sigma \sqrt{2 \ln(1 / (\eps \sqrt{\pi/2}))}\right)^2 / (2\sigma^2)}
=
\frac{\eps}{2 \sigma}.
$$
For a 2-variate normal distribution the lower bound occurs at the corner of the rectangular boundary where in the 1-variate normal distribution that passes through that point and $m=0$ the value of $x = \sqrt{2} \sigma \sqrt{2 \ln(1 / (\eps \sqrt{\pi/2}))}$.  Thus the value at the lowest point is
\begin{eqnarray*}
z^-_{\c{D}^\prime} = \varphi_{0, \sigma^2}\left(\sqrt{2} \sigma \sqrt{2 \ln (1 / (\eps \sqrt{\pi/2}))}\right)
& = &
\frac{1}{\sigma_{\text{c}} \sqrt{2\pi}} e^{-\left(\sqrt{2} \sigma_{\text{c}} \sqrt{2 \ln(1 / (\eps \sqrt{\pi/2}))}\right)^2 / (2\sigma_{\text{c}}^2)}
\\ & = &
\frac{\eps^2 \sqrt{\pi/2} }{2 \sigma_{\text{c}}} 
\\ & = & \Omega(\eps^2 / \sigma_{\text{c}}),
\end{eqnarray*}
where $\sigma_{\text{c}}^2$ is the variance in the direction of the corner.  
In the $d$-variate case, the lowest point is $\Omega(\eps^d / \sigma_{\text{c}})$.  

We calculate a bound for $\lambda_{\c{D}^\prime}$ by examining the largest second derivative along the 1-dimensional normal distribution and along an ellipse defined by the minimum and maximum variance.  In the first case we can write
$$
\frac{d^2 \varphi_{0,\sigma^2}(x)}{d x^2} 
= 
\varphi_{0, \sigma^2} \left(\frac{x^2 - \sigma^2}{\sigma^4}\right)
$$
which is maximized at $x = \sqrt{3} \sigma$.  And 
$$
\frac{d^2 \varphi_{0,\sigma^2}(\sqrt{3} \sigma)}{d x^2} = \frac{1}{\sigma^3} \sqrt{\frac{2}{\pi e^3}} = O(1/\sigma^3).
$$

For a bivariate normal distribution the largest eigenvalue of the Hessian of the extension of $\varphi_{0, \sigma^2}$ is similarly not large in the tail.  Thus our choice of $\eps$ does not effect this value.  

Hence, we can write $\left(\lambda_{\c{D}^\prime} d_{\c{D}^\prime}^2 / z^-_{\c{D}^\prime} \right) = O(\frac{1}{\eps^{2}} \log \frac{1}{\eps})$ for constant $\sigma$.
And, using Theorem \ref{thm:smooth}, we can state the following theorem.

\begin{theorem}
\label{thm:eps-normal}
For a 2-variate normal distribution $\varphi$ with constant variance, we can deterministically create an $\eps$-approximation of the range space $(\varphi, \c{R}_2 \times \b{R})$ of size $O\left(\frac{1}{\eps^4} \log^2 \frac{1}{\eps} \plle\right)$.  
This can be improved to a set of size $O(\frac{1}{\eps} \log^4 \frac{1}{\eps} \plle)$ in time $O(\frac{1}{\eps^7} \plle)$.
\end{theorem}

\end{document}